\documentclass[10pt]{article}

\usepackage{multicol}
\usepackage[a4paper,margin=2cm]{geometry}
\usepackage{amsmath,amssymb,amsfonts,amsthm}
\usepackage{bbm}
\usepackage{braket}
\usepackage[utf8]{inputenc}
\usepackage{verbatim}
\usepackage{xcolor}
\usepackage{lineno}

\usepackage[backend=bibtex,sorting=none,style=phys,biblabel=brackets,backref=false,bibencoding=utf8,maxnames=20,eprint=true]{biblatex}
\DefineBibliographyStrings{english}{%
  backrefpage = {page},
  backrefpages = {pages},
}
\usepackage[bookmarks,colorlinks,breaklinks]{hyperref}  
\definecolor{dullmagenta}{rgb}{0.4,0,0.4}   
\definecolor{darkblue}{rgb}{0,0,0.4}
\hypersetup{linkcolor=red,citecolor=blue,filecolor=dullmagenta,urlcolor=darkblue} 

\usepackage[bitstream-charter,cal=cmcal]{mathdesign}
\usepackage[T1]{fontenc}

\usepackage{titlesec}
\titleformat*{\section}{\bfseries}
\titleformat*{\subsection}{\normalsize\bfseries}
\titleformat*{\subsubsection}{\bfseries}
\titleformat*{\paragraph}{\large\bfseries}
\titleformat*{\subparagraph}{\large\bfseries}

\titlespacing\section{0pt}{12pt plus 4pt minus 2pt}{2pt plus 2pt minus 2pt}

\usepackage{tikz}
\usetikzlibrary{external}
\usepackage{pgfplots}
\pgfplotsset{width=7cm,compat=newest}
\usetikzlibrary{arrows,
    decorations.markings
}

\newcommand{\nc}{\newcommand}
\nc{\ketbra}[1]{|#1\rangle\langle #1|}
\nc{\kb}[1]{\vert#1\rangle\!\langle#1\vert}
\def\tr{\operatorname{Tr}}

\def\id{{\mathbbm{1}}}
\def\phi{\varphi}
\nc{\dd}[1]{\!\!{\rm d} #1\,}
\nc{\numberthis}{\addtocounter{equation}{1}\tag{\theequation}}
\def\channel{\ensuremath{\mathcal{E}_{k,N_0}}}

\nc{\todo}[1]{{\color{red} ToDo: #1}}

\newtheorem{lemma}{Lemma}

\usepackage{titling}
\posttitle{\par\end{center}}
\predate{}
\postdate{}
\begin{document}

\title{\large {\bf Coherent-state constellations and polar codes for thermal Gaussian channels}}
\author{
{\normalsize Felipe Lacerda$^{1,2}$, Joseph M.\ Renes$^2$, and Volkher B.\ Scholz$^{2,3}$}\\
\emph{\normalsize $^1$Department of Computer Science, Aarhus University, 8200 Aarhus N, Denmark}\\
\emph{\normalsize $^2$Institute for Theoretical Physics, ETH Z\"urich, 8093 Z\"urich, Switzerland}\\
\emph{\normalsize $^3$Department of Physics, Ghent University, 9000 Gent, Belgium}
}

\date{\vspace{-\baselineskip}}

\maketitle

\begin{abstract}
Optical communication channels are ultimately quantum-mechanical in nature, and we must therefore look beyond classical information theory to determine their communication capacity as well as to find efficient encoding and decoding schemes of the highest rates. 
Thermal channels, which arise from linear coupling of the field to a thermal environment, are of particular practical relevance; 
their classical capacity has been recently established, but their quantum capacity remains unknown. 
While the capacity sets the ultimate limit on reliable communication rates, it does not promise that such rates are achievable by practical means. 
Here we construct efficiently encodable codes for thermal channels which achieve the classical capacity and the so-called Gaussian coherent information for transmission of classical and quantum information, respectively. 
Our codes are based on combining polar codes with a discretization of the channel input into a finite ``constellation'' of coherent states. Encoding of classical information can be done using linear optics. 
\end{abstract}

\begin{multicols}{2}

\section{Introduction}

Optical communication channels such as glass fibers are the workhorses of state-of-the-art communication networks, and they are also of particular importance in quantum communication theory. First, in our quest to squeeze ever more data through existing communication infrastructure, we are gradually reaching the realm of quantum effects and hence need to consider their influence on data transmission. Second, optical channels are the most promising candidate to establish quantum links between distant locations and possibly a quantum internet. Hence developing communication schemes for such noisy quantum channels is of central importance in quantum information.   

The quantum effects of noise in such electromagnetic systems are well modeled by linear coupling of the field modes to additional Bosonic fields by quadratic Hamiltonians. 
This leads to the class of Gaussian channels, whose action can be described by linear operations in phase space~\cite{holevo_evaluating_2001,cerf_gaussian_2007}.
Mixing of a single mode with thermal noise is a particularly relevant Gaussian channel, called the thermal channel.
The ultimate capacity for transmitting classical information over thermal channels, and indeed any phase-insensitive channel, has been recently established in~\cite{giovannetti_ultimate_2014}. 
Less is rigorously established about their quantum capacity, but it is widely believed to be given by the regularized coherent information~\cite{cerf_gaussian_2007}. 
Restricting attention to Gaussian inputs, the most practically relevant case of states defined by their first and second moments, \cite{holevo_evaluating_2001} showed that thermal states with ever larger mean photon number optimize the coherent information. 
For degradable channels, e.g.\ pure loss channels in which the thermal noise of the channel is at zero temperature, Gaussian inputs to the coherent information are in fact optimal, that is, they maximize the
coherent information~\cite{wolf_quantum_2007}.

While capacity is an important property, it does not address the practical limitations of high-rate communication, such as the efficiency of encoding and decoding operations, or their implementability using linear optics.   
In this article we construct explicit codes that achieve the classical capacity as well as (likewise explicit) quantum codes that achieve the Gaussian coherent information of thermal channels. 
On the way to the latter, we construct codes for private information transmission that also achieve the coherent information.
All codes have efficient encoding operations and explicit decoders, though their decoding efficiency is unknown.
In the case of transmitting classical information, private or not, the encoding operations require only the generation of coherent states. 
Superpositions of coherent states are used for transmission of quantum information. 

Our code constructions are based on \emph{discretizing} the optimal channel input to a finite ``constellation,'' and then using a \emph{polar code} on the induced channel. 
Both are concepts originating in classical information theory. 
Finite constellations of input signals have always been used for continuous-input classical channels, e.g.\ the additive white Gaussian noise channel (AWGN), and several particular constellations are known to achieve the AWGN capacity (see~\cite{wu_impact_2010}). 
Meanwhile, polar codes are a recent breakthrough, the first explicitly constructed codes with efficient encoding and decoding that achieve the classical capacity of discrete-input channels~\cite{arikan_channel_2009}. 
Polar coding was adapted for classical communication over finite-dimensional quantum channels by Wilde and Guha~\cite{wilde_polar_2013} and for quantum communication by one of the authors~\cite{renes_efficient_2012}. 

We construct constellations for the thermal channel from AWGN constellations.
The thermal input state can be viewed as a Gaussian-weighted mixture of coherent states, i.e.\ its Glauber-Sudarshan $P$ function is a two-dimensional Gaussian.  
The optimal AWGN input is a one-dimensional Gaussian distribution, and we can therefore use AWGN constellations for the real and imaginary parts to give a finite constellation of coherent states. 
The difficulty is to show that the constellation achieves essentially the same rate as does the thermal input. 
After all, the discretized input only resembles the ideal Gaussian input in a very weak sense (specifically, in their lower-order moments), and it is not immediate that these limited similarities are enough to ensure similar rates. 
Luckily, Wu and Verd\'u~\cite{wu_impact_2010} have recently established precisely this result for the AWGN, and we adapt this to the quantum case. 
Not having to very strictly emulate the optimal input state frees us considerably in designing coding protocols and should have application to other settings, such as two-way protocols or quantum repeaters.


\begin{figure*}
\centering
\includegraphics{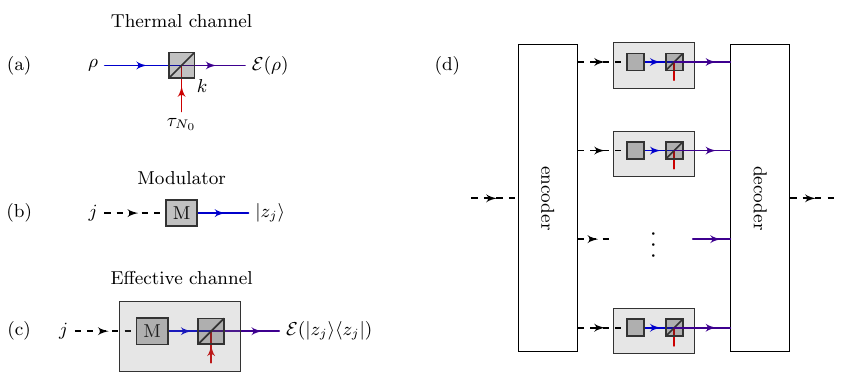}
\caption{\label{fig:channels} (a) The thermal channel $\mathcal E$ mixes the input state $\rho$ with the thermal state of mean photon number $N_0$ on a beamsplitter of transmittivity $k$. (b) The modulator taking classical input $j$ to the $j$th coherent state $\ket{z_j}$ in the constellation. (c) The effective channel resulting from combining the modulator with the thermal channel. 
(d) The classical coding scheme using the effective channel. In the quantum scheme the encoder accepts a quantum state and the modulator must be able to output superpositions of constellation coherent states.}
\end{figure*}

\section{Thermal channels}
\label{sec:bosonic-channels}

We consider single-mode thermal-noise
channels~\cite{giovannetti_ultimate_2014}, denoted $\channel$, which describes an input mode interacting with a thermal state with mean photon number $N_0$ at a beamsplitter of transmittivity $k<1$. 
The channel action can be described by the following transformation of the annihilation operators $a$ and $b$ of the input (signal) and auxiliary (ancilla) modes, respectively:
\begin{subequations}
\label{eq:channeldef}
\begin{align}
  \label{eq:thermal-1} a &\to a'=k a + \sqrt{1-k^2} b, \\
  \label{eq:thermal-2} b &\to b'=k b - \sqrt{1-k^2} a,
\end{align}
\end{subequations}
where the ancilla mode is initially in the thermal state with mean photon number $N_0$. 
Equivalently, defining $N_c=(1-k^2)N_0$, $\channel$ describes the composition of an attenuation channel ($N_0=0$ in $\channel$) with coefficient $0\leq k\leq 1$ and a channel adding Gaussian-distributed noise to each quadrature with identical variance $\sqrt{N_c}$~\cite{holevo_one-mode_2007}. 
Indeed the case $k=1$ requires the latter formulation; see~\cite[Sec.\ 4]{holevo_one-mode_2007}.

For a fixed maximum input mean photon
number $N$, the classical capacity $C(N)$ of the thermal-noise channel has been
shown~\cite{giovannetti_ultimate_2014} to be the Holevo information of the channel, evaluated for
a Gaussian ensemble of mean photon number $N$ of coherent states.

More precisely, we have the following. 
Let $Z$ be a random variable over $\mathbb C$ with probability density $P_Z(z)=\psi_N(z)$ for 
\begin{align}
\label{eq:Zdistrib}
\psi_N(z):=(\pi N)^{-1} \exp(-|z|^2/N),
\end{align}
$\ket{z}$ the coherent state centered at $z\in \mathbb C$, and $\theta_z^B=\channel^{A\to B}(\ketbra{z}^A)$. Here $A$ denotes the input mode and $B$ the output mode. 
We denote by $E$ the output ancilla mode, so that the joint output produced by the channel is the bipartite state $\theta_z^{BE}$. 
The capacity can then be expressed as~\cite{giovannetti_ultimate_2014}
\begin{equation}
  \label{eq:classical-capacity}
  C(N) =  I(Z:B)_\rho,
\end{equation}
where $\rho^{ZB}$ is the classical-quantum state corresponding to the ensemble $\{P_Z(z),\theta_z^B\}_{z\in \mathbb C}$, 
and $I(Z:B)_\rho$ is defined as
\begin{align}
I(Z:B)_\rho=H(\rho^B)-\int_{\mathbb C} \text{d}z\, P_Z(z) H(\theta_z^B),
\end{align}
with $\rho^B=\int_{\mathbb C} \text{d}z\, P_Z(z)\theta_z^B$.
Here $H(\rho)$ is the von Neumann entropy, $H(\rho)=-\tr[\rho\log\rho]$, and we use the natural logarithm throughout.  
The marginal state of the input is then just the thermal state $\tau_N$ having average photon number $N$: $\int_{\mathbb C}\text{d}z\, P_Z(z)\kb z=\tau_N$, with 
\begin{equation}
  \label{eq:thermal-state}
  \tau_N := \frac{1}{N+1} \sum_{n=0}^{\infty}
  \left(\frac{N}{N+1}\right)^n \kb{n},
\end{equation}
where $\{\ket{n}\}_{n=0}^{\infty}$ is the number basis. 
Observe that $P_Z$ is the $P$ function of the input state $\tau_N$, while for given $z$ the above description of the channel implies that the $P$ function of the output $\theta_z^B$ in $B$ is simply $P_{\theta_z^B}(w)=\psi_{N_c}(w-kz)$. 
Similarly, the output $\theta_z^{E}$ in the auxiliary mode has $P$ function $P_{\theta^E_z}(w)=\psi_{k^2N_0}(w+\sqrt{1-k^2}z)$.

In establishing the formula for the classical capacity, \cite{giovannetti_ultimate_2014} relies on random coding arguments. For the pure loss case ($N_0=0$), Guha and Wilde  constructed polar codes for a constellation of two coherent states, known as binary phase shift keying (BPSK)~\cite{guha_polar_2012}, whose optimal rate approaches the capacity in the limit of vanishing input mean photon number. \\

Compared to the case of classical data transmission, less is known about the best rate to transmit \textit{quantum data}, which is set by the \textit{quantum} capacity of the thermal noise channel (we refer the reader to the book~\cite{wilde_quantum_2013} for the precise definitions and an introduction into these concepts). This is related to the fact that we currently do only have a limited understanding of quantum coding in the infinite-dimensional setting, especially when taking physical limitations such as precision restrictions into account. For example, Devetak's proof~\cite{devetak_private_2005} that the coherent information and its regularization are achievable rates of quantum communication does not deal with infinite-dimensional channels. Although it seems straightforward to extend his method by suitably truncating the input and output spaces (as carried out for entanglement distillation in~\cite[Appendix I]{pirandola_fundamental_2017}), this proof technique would lead to random codes involving superpositions of products of number states. These encodings are argueably very hard to realize experimentally. 

It is thus interesting to determine possible rates of noiseless communication if we restrict to \textit{Gaussian} encodings. Adding another simplification, namely, disregarding the fact that the coherent information needs to be regularized in order to give the ultimate quantum capacity of a channel, leads to single-letter Gaussian quantum capacity $Q_G^{(1)}$. This quantity is given by maximizing the coherent information over Gaussian input states,
\begin{align}
  \label{eq:coherentinfo}
  Q_G^{(1)} = \max_{\rho} H(B)_\omega-H(E)_\omega\,,
\end{align}
where $\omega^{ABE} = V_{\channel}^{A' \to BR}(\xi^{AA'})$ for $\xi^{AA'}$ a purification of the Gaussian state $\rho^A$ at the channel input and $V$ is a Stinespring dilation of the channel. 
In~\cite{holevo_evaluating_2001}, Holevo and Werner showed that this optimization problem can be solved explicitely in the case of the thermal noise channel, and that the unique maximizer is the state $\tau_N$ for $N\to \infty$. 
The corresponding output is given by
$\tau_{N'}$ where $N' = k^2 N + N_c$. 
They also showed that $Q_G^{(1)}$ remains finite in this limit, that is, in the limit of infinite energy. 
This is in sharp constrast to the classical case, where the capacity is infinite in the absence of an energy bound. 
The difference between the classical and the quantum case may be argued for by the fact that the quantum uncertainty principle prevents us from arbitrary dense packing of quantum information, but a satisfactory analytic understanding of this point is still missing. 
In addition, a finite-energy bound is of course also of practical interest, even in the quantum-mechanical case, due to implementation limitations.

Nevertheless, the quantity~\eqref{eq:coherentinfo} is a lower bound on the quantum capacity, because it is restricted to Gaussian inputs which are moreover not entangled in the case of multiple channel uses. Hence, the construction of communication schemes which achieve this rate is a necesssary first step for advancing our understanding of the limitations of quantum communication via the thermal noise channel. Moreover, in the case of degradable channels such as pure loss, Wolf \emph{et al.}\ have shown that Gaussian inputs are provably optimal among all inputs~\cite{wolf_quantum_2007} (and moreover entangled inputs are not necessary~\cite{devetak_capacity_2005}). 

The first communication schemes which achieve the rate~\eqref{eq:coherentinfo} were constructed by Harrington and Preskill for the case $k=1$ in~\cite{harrington_achievable_2001}. Their construction is not based on discretization, but rather embedding an appropriate number of qubits directly into the state space of a larger number of modes, using the method of~\cite{gottesman_encoding_2001}. The practicality of this method is, however, limited, even disregarding the non-explicit nature of the code, as it ostensibly requires codewords which are superpositions of highly squeezed states.

Here, we construct new encoding schemes which are based on coherent states and are thus more feasible experimentally. As already explained in the introduction, our results are based on constellations, or discretizations of continuous input distributions.

\section{New coding schemes}
Any given constellation defines a mapping, or modulation, from a discrete set of $m$ inputs (indexing the particular input state) to the coherent states in the constellation. 
Combining this mapping with the thermal channel then defines an ``effective'' discrete-input Bosonic-output channel, as depicted in Figure~\ref{fig:channels}.
The methods of polar coding can then be applied to this channel to yield a high-rate block code. 
As we show in the two subsequent sections, this results in classical codes with rates achieving the capacity $C(N)$ and quantum codes with rates achieving the Gaussian coherent information $Q^{(1)}_G$. 
The use of polar codes also ensures the encoding operation is efficient, though it is not known if efficient decoding is possible when the channel outputs noncommuting states.

In classical information theory, channel constellations are designed to emulate the
optimizer of the channel mutual information, as this gives a means for showing
that the overall coding scheme can approach the capacity. 
Here we take the same approach, and aim to emulate $\tau_N$, the optimizer in both the Holevo information and coherent information. 

As depicted in Figure~\ref{fig:constellations}, several useful constellations are known for the AWGN. 
The equilattice is simply $m$ equally spaced points with equal probability whose variance matches that of the optimizer of the channel mutual information~\cite{ungerboeck_channel_1982}. The quantile constellation chooses points based on the quantile (inverse of the cumulative distribution)~\cite{sun_approaching_1993}, while the random walk is precisely the distribution of positions of a suitably rescaled random walk of length $(m-1)$. 
Finally, the Gauss-Hermite constellation is based on Gauss-Hermite quadrature. 
This discretizes the Gaussian distribution to $m$ points in $\mathbb R$ such that the first $2m-1$ moments of the Gaussian are correctly reproduced. 
Indeed, Gauss-Hermite quadrature is optimal in that it reproduces the largest number of moments for fixed $m$. 
This constellation and the random walk were introduced in~\cite{wu_impact_2010}. 

As noted in~\cite{wu_impact_2010}, the equilattice has a minimal gap to capacity of roughly $0.25$ bits, while the gap closes in the $m\to \infty$ limit for the quantile, random walk, and Gauss-Hermite constellations. The latter closes exponentially in $m$ in the asymptotic limit, but only polynomially in the other two cases. 
Nonetheless, the capacity of the random-walk constellation is already quite close to the optimal value even for modest $m$, while the Gauss-Hermite constellation only becomes optimal for larger $m$. 
We observe the same behavior for the thermal channel, as depicted in Figure~\ref{fig:rates}.

\begin{figure*}
\centering
\includegraphics{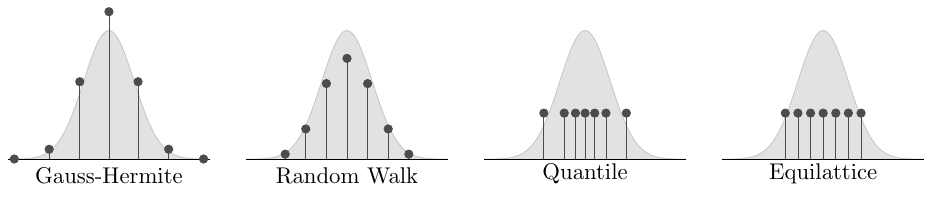}
\caption{\label{fig:constellations}
Four constellations of $m=7$ points which approximate a normally distributed random variable. 
The random-walk and equilattice constellations consist of equally spaced points, while the quantile and equilattice have equal probabilities. Gauss-Hermite quadrature  has neither, but precisely reproduces the largest number of moments of the Gaussian distribution, the first $2m-1$. The others have only the same mean and variance. 
}
\end{figure*}

\begin{figure*}
\centering
\includegraphics{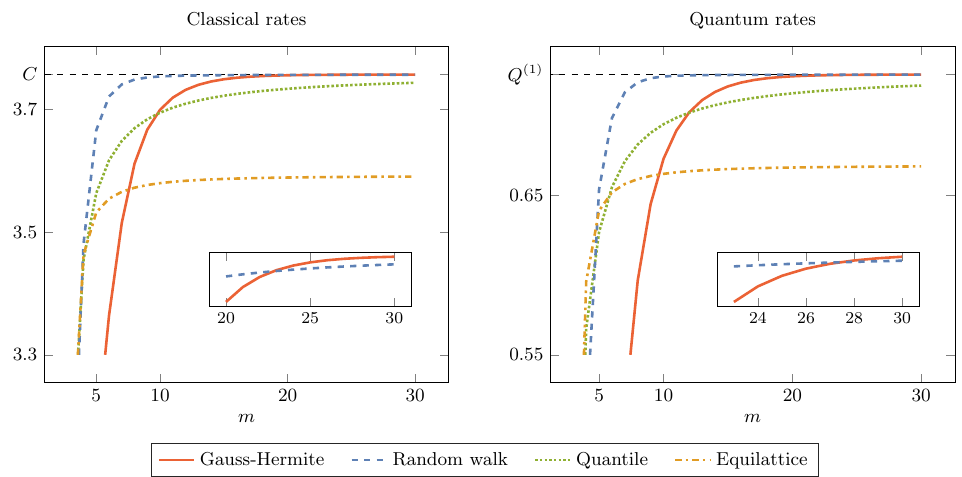}
\caption{\label{fig:rates} Achievable rates for classical and quantum information transmission with increasing constellation size $m$, for the pure loss channel ($N_0=0$) with $k=0.8$ and input state $\tau_N$ with $N=7$. Although the Gauss-Hermite constellation (straight line) is optimal as $m\to \infty$, it is inferior for small $m$. For practical purposes, the random-walk constellation (dashed line) is best, rising very quickly to rates quite close to capacity.}
\end{figure*}

\section{Coherent state constellations}
\label{sec:approaching-coherent}

For the thermal channel we can apply any of these constellations in phase space, to the $P$ function of $\tau_N$. 
As described above, this is a circularly symmetric complex Gaussian, and we use the AWGN constellations to mimic the real and imaginary parts separately.  
Specifically, the constellation is described by the distribution $Q_{N,m}(z)$, supported on $m^2$ points, such that 
\begin{align}
\label{eq:Qmdef}
\tfrac N2Q_{N,m}(\sqrt{\tfrac N2}(x+iy))=P_{X_m}(x)P_{X_m}(y),
\end{align} 
where $P_{X_m}$ is one of the four constellations considered
in~\cite{wu_impact_2010}. In terms of random variables,
$Z_m=\sqrt{\tfrac N2}(X_m+i X'_m)$, where $X_m$ and $X_m'$ are independent
realizations of the given constellation.  The factor $\sqrt{N}$ ensures that the
resulting $P$ function has variance $N$, while the $1/\sqrt{2}$ factor takes
care of the conversion from two real to one complex variable.

Define the associated ensemble
\begin{align}
\label{eq:modulatorinput}
\rho_{N,m}^{ZA}=\{Q_{N,m}(z),\kb z\}_{z\in \mathbb C}.
\end{align}
As for the classical Gaussian channel, for $m \geq 2$ the first two moments of
$\rho^A_{N,m}$ match those of the thermal state $\tau_N$.

Now we sketch how to upper bound the gap between the coherent information and
the rate of the polar code applied to the induced discretized channel.  A
simpler version of the same argument yields the analogous result for the
classical capacity.

Let us denote the target rate without using modulation as $R$, and the
associated rate using the modulation with $m$ points $R_m$. The former is the
coherent information of the channel, while the latter is the same entropic
expression, but evaluated for the input state given by the discretization
scheme.  
The coherent information may be written as $R=I(Z:B)-I(Z:E)$ for $P_Z$ as in \eqref{eq:Zdistrib}, and
$R_m$ is the same expression, evaluated with $Z\sim Q_{N,m}$.
Here $E$ is the output of the channel to the environment. 
This follows because $\theta_z^{BE}$, the state of $BE$ given the input $Z=z$ (or $Z_m=z$), is pure and hence $H(B|Z)+B(E|Z)=0$. 

We would like to find an upper bound on $\Delta=R-R_m$, which can be written as $\Delta=\Delta_B-\Delta_E$ for $\Delta_B=I(Z:B)-I(Z_m:B_m)$ and similarly for $\Delta_E$.  
Here, $B_m$ is the channel output for the input $Z_m$. 
Clearly, we need not consider $\Delta_E$ for the problem of classical information transmission. 
These quantities can be written as relative entropies. 
More specifically, 
let $\rho^{ZBE}$ and $\rho_{m}^{Z_m B_m E_m}$ be ensembles of the state
  $\theta_z^{BE}$ with distributions $\psi_N$ as in \eqref{eq:Zdistrib} and $Q_{N,m}$ as in \eqref{eq:Qmdef}, respectively.
  Then 
\begin{subequations}
\begin{align}
\Delta_B &:= I(Z : B) - I(Z_m :B_m)= D(\rho^{B_{m}}_m \| \rho^B),\label{eq:DeltaB}\\
\Delta_E &:= I(Z:E) - I(Z_m : E_m)=D(\rho^{E_m}_m \| \rho^E),
\end{align}
\end{subequations}
where the relative entropy is defined as $D(\rho\|\sigma)=\tr[\rho(\log\rho-\log\sigma)]$.

To see this, consider the claim for $\Delta_B$; the following argument will also work for $\Delta_E$.  
Expanding out the mutual information, we obtain $\Delta_B = H(B) - H(B \vert Z) - H(B_m) + H(B_m \vert Z_m)$. 
Since $P$ functions of the output states $\theta_z^B$ are all identical up to translation, it follows that their entropies are also identical, and therefore $H(B \vert Z) = H(B_m \vert Z_m)$.
This leaves $\Delta_B=H(B)-H(B_m)$. 
Using the form of the relative entropy, it is apparent that the claim is equivalent to the statement $\tr[ \rho_m^{B_m} \log \rho^B] = \tr
  [\rho^B \log \rho^B]$.
Since the channel is Gaussian, $\rho^B$ 
is a zero-mean Gaussian state, equivalent up to symplectic unitary conjugation to a tensor product of thermal states~\cite{holevo_evaluating_2001,cerf_gaussian_2007}. 
It follows that $\log \rho^B$ is a second-order polynomial $q(r,r^{\dagger})$ of the output creation and annihilation operators $r$ and $r^\dagger$. 
We thus have
\begin{subequations}
  \begin{align}
    \tr [\rho^B \log \rho^B] 
    &= \tr [\channel^{A\to B}(\rho^A) q(r,r^{\dagger}) ]\\
    &= \tr [\rho^A \channel^{A\to B\dagger}(q(r,r^{\dagger})) ]\\
    &= \tr [\rho^A p(a,a^{\dagger})],
  \end{align}
  \end{subequations}
  where $\channel^\dagger$ is the channel adjoint and, since the channel is Gaussian, $p(a,a^{\dagger})$ is a second-degree polynomial in the input creation and annihilation operators. 
  Similarly, we have $\tr [\rho_m^{B_m} \log \rho^B] = \tr[\rho_m^A p(a,a^{\dagger})]$. 
  Finally, since the first two moments of $\rho_m^A$ match those of $\rho^A$, this also holds for the outputs $\rho_m^B$ and $\rho^B$, and we indeed have $\tr[ \rho_m^{B_m} \log \rho^B] = \tr
  [\rho^B \log \rho^B]$. 


As the relative entropy is non-negative, we immediately have $\Delta_E \geq
0$, and we need only an upper bound for $\Delta_B$. 
A convenient choice is the $\chi^2$ distance, given by $\chi^2(\rho, \sigma) = \tr
\left[(\rho \sigma^{-1/2})^2 \right] - 1$, since $D(\rho \| \sigma) \leq
\chi^2(\rho, \sigma)$~\cite{ruskai_convexity_1990,temme_2-divergence_2010}.

In~\cite{wu_impact_2010} Wu and Verd\'u obtain a bound on the $\chi^2$ distance
between the optimal output state and the discretized version for the case of the
AWGN. They relate $\chi^2(P_{Y_m},P_Y)$ directly to the
moments of $X_m$, specifically the moments of the Hermite polynomials of $X_m$.
(These all vanish for standard normal $X$, save the zeroth-order polynomial.)
Indeed, they show the following slightly more general statement~\cite[\S V.A]{wu_impact_2010}. 
First denote the AWGN with signal-to-noise ratio $s$ by $W_s$; it has the action
$W_s(X)=\sqrt{s}X+G$, where $G$ is a normally distributed random variable with unit variance. 
Then for $X\sim \mathsf{N}(0,1)$, $X'$ an arbitrary random variable with density $P_{X'}$, and $Y=W_s(X)$ and $Y'=W_s(X')$, 
\begin{align}
\label{eq:wuverduchi}
1+\chi^2(P_{Y'},P_Y)= \sum_{k=0}^\infty \frac1{k!}\left(\frac s{1+s}\right)^k \left|\mathbb E[H_k(X')]\right|^2.
\end{align}

We can use their result to bound $\chi^2(\rho_m^{B_m},\rho^B)$. 
In particular, with $Q_{N,m}(\sqrt{\frac N2}(x+iy))$ as in \eqref{eq:Qmdef}, $s= k^2 N/(\sqrt{N'(N'+1)}-k^2N)$, and $Y_m=W_s(X_m)$,
\begin{align}
\label{eq:DeltaBbound}
\Delta_B\leq (1+\chi^2(P_{Y_m},P_Y))^2-1.
\end{align}
This bound is a consequence of the $\chi^2$ upper bound on the relative entropy described above, and finding tractable expressions for quantum and classical $\chi^2$ quantities relies on the fact that the second argument to the relative entropy is a thermal state or a Gaussian probability distribution, respectively. 
The inequality \eqref{eq:DeltaBbound} is an instance of a more general statement relating the quantum and classical $\chi^2$ quantities in this setting for general factorizable $P$ functions. 
The proof involves somewhat intricate Gaussian integration and is given in the Appendix.

Hence, per~\cite[\S VII]{wu_impact_2010}, using the quantile or random-walk
constellations will lead to $\Delta_B$ decaying as the inverse of the number of
points in the coherent-state constellation $m^2$, while the Gauss-Hermite
quadrature leads to $\Delta_B\leq O(e^{-c m})$ for $c=2\log \frac{1+s}{s}$. 
Using the above expression for $s$, one obtains $c\approx 2\frac{1-k^2}{k^2}\frac{N_0}{N}$ for $N\gg N_0$.
If we wish to increase $N$ but fix the gap to capacity or coherent information, this implies that the number of constellation points must scale linearly with $N$. 

\section{Polar codes for the thermal channel}
\label{sec:polar}

The polar code construction can be used for classical, private classical, and
quantum coding. Here we provide only the essential details for the present case and will not attempt to provide a background on polar coding. 
One could also appeal to other capacity-achieving schemes for the task of transmitting classical information, e.g.\ spatially-coupled low density parity check (LDPC) codes~\cite{kudekar_spatially_2013} (after showing these work for channels with quantum output), but no explicit quantum coding schemes besides polar codes are known to achieve the coherent information. 
Common to all our coding scenarios is a truncation of the \emph{output} space to finite dimensions, as described at the end of \S\ref{sec:bosonic-channels}.
This ensures that we may apply existing results on the construction and properties of polar codes. {It appears that all the necessary statements also hold for channels with infinite-dimensional outputs, but we have not shown this definitively.}

\subsection{Classical coding}

Wilde and Guha~\cite{wilde_polar_2013} show how to construct polar codes to transmit classical information over channels with binary classical input and finite-dimensional quantum output. 
Recently, one of the present authors extended the construction to arbitrary input alphabets~\cite{nasser_polar_2017}.
Here, the modulator selects a coherent-state input for the Bosonic Gaussian channel, but the choice of which input to make is classical. 
Both \cite{wilde_polar_2013} and \cite{nasser_polar_2017} only considered the case of uniformly random channel input (befitting the equilattice or quantile constellations), but this restriction can be lifted by the construction of Honda and Yamamoto~\cite{honda_polar_2013} (using modified polarization statements found in \cite{renes_efficient_2015}), so as to apply to random-walk or Gauss-Hermite constellations. 

Thus, a rate of $I(Z_m:B_m)$ is achievable. 
Both \cite{wilde_polar_2013} and \cite{nasser_polar_2017} appeal to Arikan's original encoder, hence encoding is
efficient. 
The quantum version of the successive cancellation decoder is
likewise explicit, but its implementation complexity is unknown. 

\subsection{Private coding}

For private coding, the task is not only to transmit information reliably, but
also to hide it from an eavesdropper. 
Codes to do so are constructed by Renes and Wilde in~\cite{renes_polar_2014}, 
where it is shown that the ``naive'' wiretap rate $I(Z_m : B_m) - I(Z_m : E_m)$ is achievable. 
Generally, the code requires secret-key assistance at nonzero rate, but not if the channel is degradable~\cite{wolf_quantum_2007}, e.g.\ for pure loss ($N_0=0$). 

Crucial to the construction is the observation is that the channels to the legitimate receiver Bob and to the
eavesdropper Eve are related by an entropic uncertainty relation. 
Eve's information about the actual classical message must be small if the code is constructed so that it could send complementary ``phase'' information to Bob. The construction in~\cite{renes_polar_2014} is again for uniform channel inputs, but can be extended as just described above. The uncertainty principle argument is unaffected by having non-uniform channel inputs.

\subsection{Quantum coding}

For quantum coding, we proceed as in the case of private coding, except we use
the channels to send quantum states. Now we regard the modulator as mapping the $k$th basis state of an 
$m^2$-dimensional quantum system to the $k$th coherent state in the constellation. 
Since this mapping is not unitary, the coherent information of the modulator composed with $\channel$ is lower than that of $\channel$ itself. This implies that we cannot employ the method of Devetak~\cite{devetak_private_2005} to upgrade a private code to a quantum code (at least, not at the same rate). 

Instead, we employ a scheme based on one-way entanglement distillation, combining it with teleportation to enable transmission of arbitrary states~\cite{bennett_mixed-state_1996}. 
Recall that in entanglement distillation, 
Alice and Bob use local operations and classical communication to transform many copies of
a bipartite mixed state into a number of maximally entangled states. 
The mixed state in question is that obtained by Alice transmitting the state in \eqref{eq:modulatorinput} through the channel to Bob, while keeping its purification. 
Abusing notation and denoting by $Q_{N,m}(j)$ the probability of the $j$th coherent state $\ket{z_j}$, the state at the input to the channel can be written
\begin{align}
\ket{\xi}^{A'A}=\sum_{j=1}^{m^2}\sqrt{Q_{N,m}(j)}\ket{b_j}^{A'}\ket{z_j}^A,
\end{align} 
where the $\ket{b_j}$ are an orthonormal basis for an $m^2$-dimensional space $A'$. 
Here we are interested in using stabilizer codes for entanglement distillation, 
where Alice makes stabilizer measurements on $A'$ and sends the outcomes
(the syndrome) to Bob. Bob then uses this side information to execute a decoding
operation on his system. By choosing a suitable stabilizer code, Alice and
Bob end up with copies of a maximally entangled state. As shown
in~\cite{renes_efficient_2012,renes_efficient_2015}, by using a polar code the scheme has rate equal
to $I(Z_m:B)-I(Z_m:E)$ which approaches $Q_G^{(1)}(\channel)$ as $m\to \infty$. 
The former construction is simpler, yet may require entanglement assistance, while the latter is somewhat more complicated but does not require entanglement assistance. 
In either case Alice's stabilizer measurements can be done efficiently, since the requisite quantum circuit is just the polar coding circuit. 

This scheme can be converted into a quantum code involving no classical communication as follows. 
Consider a particular syndrome, selected in advance and known to Bob (since the
final state is entangled, the distribution of syndromes is essentially uniform).
Instead of preparing many copies of the state $\ket{\xi}^{A'A}$ and making the
stabilizer measurement, Alice could just create the bipartite state that results
when the given syndrome is observed. After receiving the channel output Bob
proceeds with his decoding operation. By the properties of stabilizer codes,
each logical codeword will correspond to a superposition of tensor products of
coherent states. 
The catch in this construction is that we no longer have any guarantee that the codewords are efficiently constructable except by the tedious protocol of performing the procedure above and keeping the state only when the desired syndrome is obtained.

\section{Conclusions and Outlook}
We have constructed three classes of codes for the thermal noise channel, all based on concatenation of polar coding with suitable discretizations of the channel into a constellation of input coherent states. 
For transmitting classical information, this leads to explicit codes with efficiently implementable encoders that achieve the single letter channel capacity if restricted to Gaussian inputs. 
Moreover, the encoder need only prepare products of coherent states. 
For transmitting classical information privately, our codes naively achieve the unoptimized wiretap rate $I(Z:B)-I(Z:E)$, but this could presumably be improved by preprocessing exactly as in the classical wiretap scenario. 

In the fully quantum case, we have shown how to employ a one-way entanglement distillation scheme in order to send quantum information at a rate given by the Gaussian coherent information. Alice's operations are efficiently implementable in this scheme, but it is not clear how to efficiently generate the corresponding quantum codewords for use in a standard error-correction scenario. From a more technical perspective, this construction provides a rigorous proof taking the subtleties of the infinite-dimensional setting into account that the Gaussian coherent information is a lower bound on the quantum capacity of the thermal channel and equal to it for pure loss. 

While our results thus provide explicit and efficiently implementable encodings, the question of how to construct efficient and experimentally realizable decoders is still wide open. Although an explicit decoder is known for classical information transmission, the successive cancellation decoder of~\cite{wilde_polar_2013}, it is not known how to implement it efficiently for non-commuting channel outputs. Moreover, recent results by Winter and co-workers even suggest that, while Gaussian encodings are sufficient to achieve the capacity, Gaussian decoders seem not to be sufficient~\cite{Winter_2017}. 
This again shows that much further work is needed to understand the limitations of data transmission through Gaussian quantum channels.

\section*{Acknowledgments} 
We thank Saikat Guha and Mark M.\ Wilde for helpful discussions on the problems of constructing codes for Gaussian channels. 
This work was supported by the Swiss National Science Foundation (through the
National Center of Competence in Research ``Quantum Science and Technology'' and
Grant No. 200020-135048) and by the European Research Council (ERC Grant No.\
258932). V.B.S. also acknowledges support from the ERC (Grant No.\ 197868) and from
an ETH Postdoctoral Fellowship. F.L. also acknowledges support from the Danish
National Research Foundation and The National Natural Science Foundation of China (under
the grant 61361136003) for the Sino-Danish Center for the Theory of Interactive
Computation and from the Center for Research in Foundations of Electronic
Markets, supported by the Danish Strategic Research Council.


\appendix

\section{The relative entropy upper bound}
Equation~\ref{eq:DeltaBbound} follows from \eqref{eq:DeltaB}, the upper bound $D(\rho||\sigma)\leq \chi^2(\rho,\sigma)$, and the following 
\begin{lemma}
\label{lem:chi-P}
Let $X \sim \mathsf{N}(0,1)$ and $X'_1, X'_2$ be arbitrary real random variables with  densities $P_{X'_1}$ and $P_{X_2'}$. 
For $Z=\sqrt{\tfrac N2}(X'_1+i X'_2)$ with density $Q(z)$ and arbitrary $N>0$,
define $\rho=\int_{\mathbb C}\,\,\dd z Q(z) \, \theta_z$. 
Then, with $Y=W_s(X)$ and $Y'_j=W_s(X'_j)$, $N' = k^2 N + N_c$, and $s = k^2 N/(\sqrt{N'(N'+1)}-k^2N)$,
\begin{align*}
1+\chi^2(\rho,\tau_{N'})=(1+\chi^2(P_{Y'_1},P_Y))\,(1+\chi^2(P_{Y'_2},P_Y)).
\end{align*}
\end{lemma}

We prove Lemma~\ref{lem:chi-P} by establishing two intermediate results which give explicit expressions for the classical and quantum $\chi^2$ quantities. 
First, define the real Gaussian density 
\begin{align}
  \label{eq:gaussian}
  \phi_{s}(x) &:= \frac{1}{\sqrt{2 \pi s} } \exp \left[-\frac{x^2}{2
      s} \right],
\end{align}
for $x \in \mathbb R$ and $s>0$; recall that the corresponding complex density $\psi_s$ is defined in \eqref{eq:Zdistrib}. 
Observe that $\phi_s(x)=\tfrac 1{\sqrt s}\phi_1(\tfrac x {\sqrt s})$ and  $\psi_s(x+iy)=\phi_{s/2}(x)\phi_{s/2}(y)$.
We will make use of the following Gaussian integral formula. For $A$ an arbitrary $n\times n$ complex matrix with positive definite Hermitian part,
i.e., $\frac{1}{2} (A + A^{\dagger}) \geq 0$, and arbitrary $\mathbf u,\mathbf v\in \mathbb C^n$ we have~\cite[Eq.\ 3.18]{altland_condensed_2010}.
\begin{equation}
\label{eq:gaussian_integral}
\int_{\mathbb{C}^n}\!\!{\rm d}\mathbf w\,\, e^{-\bar{\mathbf w}^TA \mathbf w+\bar{\mathbf u}^T \mathbf w+\bar{\mathbf w}^T \mathbf v}=\pi^n\det A^{-1}e^{\bar{\mathbf u}^T A^{-1}\mathbf v}.
\end{equation}

Now we give the expression for the relevant classical $\chi^2$ quantity, $\chi^2(P_{Y'},P_Y)$ with an arbitrary input $X'$. 
To this end define, for $s\geq 0$ and $x,x\in \mathbb R$,
\begin{align}
\label{Ksdef}
K_s(x,x'):=
\frac{1{+}s}{\sqrt{1{+}2s}}\exp\left[-\frac{s}{2(1{+}2s)}\left(s(x{-}x')^2-2xx'\right)\right].
\end{align}
\begin{lemma}
\label{lem:classicalchi2expression}
For random variables $X$, $X'$, $Y$, and $Y'$ as in~\eqref{eq:wuverduchi}, 
\begin{align}
1+\chi^2(P_{Y'},P_Y)= \int_{\mathbb R^2}\!\!\dd x\,\,\dd{x'} P_{X'}(x)P_{X'}(x') K_s(x,x').
\end{align}
\end{lemma}
\begin{proof}
First observe that $P_Y(y)=\phi_{1+s}(y)$, while 
\begin{equation}
P_{Y'}(y)=\int_{\mathbb R} \dd x\, P_{X'}(x)\phi_{1}(y-\sqrt{s}x).
\end{equation}
Computing $1+\chi^2(P_{Y'},P_Y)$, we find  
\begin{align}
 1+\chi^2(P_{Y'},P_Y)
 &=\int_{\mathbb R} \dd y \frac{P_{Y'}(y)^2}{P_Y(y)}\\
 &=\int_{\mathbb R}\dd x\!\!\int_{\mathbb R}\dd{x'} P_{X_m}(x)
P_{X'}(x') I_s(x,x')\,,
\end{align}
where 
\begin{align}
I_s(x,x')=\int_{\mathbb R}\dd y \phi_{1+s}(y)^{-1} \phi_{1}(y-\sqrt{s}x)\phi_{1}(y-\sqrt{s}x')
\end{align}
This is a simple Gaussian integral, and using \eqref{eq:gaussian_integral} we find $I_s(x,x')=K_s(x,x')$. 
\end{proof}

Next we turn to the expression for the quantum $\chi^2$ quantity. 
\begin{lemma}
\label{lem:chi2generalcalculation}
Let $\rho$ be a state with positive $P$ function (a probability density $P$). 
For any $N> 0$, 
\begin{align}
\label{eq:chi2gen}
1+\chi^2(\rho,\tau_N)=\int_{\mathbb{C}^2} \dd z\,\, \dd {z'} P(z)\, P(z')\, C_N(z,z'),
\end{align}
where, for $t_N=\sqrt{\frac{N+1}N}$, 
\begin{align*}
C_N(z,z')
&:=(N+1)\exp\left[-|z|^2-|z'|^2+t_n(z\bar z'+\bar zz')\right].
\end{align*}
\end{lemma}

\begin{proof}
Computing the trace in the number basis, we get
\begin{align*}
1{+}\chi^2(\rho,\tau_{N})
&= \sum_{n=0}^\infty \bra{n}\rho\,\tau_{N}^{-1/2}\rho\, \tau_{N}^{-1/2}\ket{n}\\
&= (N{+}1)\sum_{n,n'=0}^\infty t_N^n
  t_N^{n'}\bra{n}\rho\ket{n'}\bra{n'}\rho\ket{n} \\
&=(N{+}1)\!\!\sum_{n,n'=0}^\infty \!t_N^n t_N^{n'} \!\int_{\mathbb{C}}\dd z' P(z') \braket{n'|z'}\braket{z'|n}\\
&=\int_{\mathbb{C}^2}\dd z\,\,\dd z' P(z) P(z') S_N(z,z')\,,
\end{align*}
where 
\begin{align*}
S_N(z,z')
&=(N+1)\!\!\sum_{n,n'=0}^\infty t_N^n t_N^{n'}\braket{n|z}\braket{z|n'}\braket{n'|z'}\braket{z'|n}\,.
\end{align*}
Computing $S_N(z,z')$, we find 
\begin{align*}
\frac{S_N(z,z')}{N+1}
&= \sum_{n,n'=0}^{\infty} t_N^n t_N^{n'} e^{-|z|^2}
\frac{z^{n}\bar{z}^{n'}}{\sqrt{n!} \sqrt{n'!}} e^{-|z'|^2} \frac{z'^{n'}\bar{z}'^{n}}{\sqrt{n'!} \sqrt{n!}} \\
&=  e^{-|z|^2-|z'|^2} \sum_{n=0}^{\infty} \frac{1}{n!} \left(t_N{z} \bar z' \right)^n \sum_{n'=0}^{\infty} \frac{1}{n'!}\left( t_N {z}' \bar z \right)^{n'} \\
&=\frac{C_N(z,z')}{N+1}.
\end{align*}
This completes the proof.
\end{proof}

\end{multicols}
With these two intermediate results, we are ready to establish Lemma~\ref{lem:chi-P}.
\begin{proof}[Proof of Lemma~\ref{lem:chi-P}]
Using the form of $\theta_z$ as described in Sec.~\ref{sec:bosonic-channels}, it is apparent that the $P$ function of $\rho$ is simply
\begin{align}
P_{\rho}(w)=\int_{\mathbb C}\,\dd z Q(z)\psi_{N_c}(w-kz).
\end{align}
Now applying Lemma~\ref{lem:chi2generalcalculation} with number parameter $N'$, we have
\begin{align}
1+\chi^2(\rho,\tau_{N'})
&=\int_{\mathbb{C}^2} \dd z\,\, \dd {z'} Q(z)\, Q(z') R_{N'}(z,z')\,,\label{eq:QQR}
\end{align}
where $R_{N'}(z,z')$ is just 
\begin{align*}
\int_{\mathbb{C}^2}\!\! \dd w\,\, \dd {w'} \psi_{N_c}(w{-}kz) \psi_{N_c}(w'{-}kz') C_{N'}(w,w')\,.
\end{align*}

Recalling \eqref{eq:Zdistrib} and using $A_{N} =
  \begin{pmatrix}
    1 & t_N \\
    t_N & 1
  \end{pmatrix}$ to express the argument to the exponential in $C_{N'}(w,w')$ in matrix form, we obtain a Gaussian integral for $R_{N'}(z,z')$:

\begin{equation}
  \label{eq:R_N}
  R_{N'}(z,z')=\frac{N'+1}{\pi^2
    N_c^2} e^{-\frac{k^2}{N_c}(|z|^2+|z'|^2)}\int_{\mathbb C^2} \dd{\mathbf
    w}\exp\left[-\bar{\mathbf w}^T A_{N'} \mathbf w-\frac{1}{N_c}\bar{\mathbf
      w}^T\mathbf w+\frac{k}{N_c}(\bar{\mathbf z}^T\mathbf w+\bar{\mathbf
      w}^T\mathbf z) \right].
\end{equation}

Here $\mathbf z=(z,z')$ and similarly for $\mathbf w$. 
The integrand has the form $e^{-\bar{\mathbf w}^TA \mathbf w+\bar{\mathbf u}^T
  \mathbf w+\bar{\mathbf w}^T \mathbf v}$ where $A = A_{N'} + \id/N_c$ and
$\mathbf{u} = \mathbf{v} = k/N_c
\mathbf{z}$. Applying~\eqref{eq:gaussian_integral} to~\eqref{eq:R_N}, we obtain
\begin{align}\label{eq:defRkernel}
R_{N'}(z,z')&=\frac{N'+1}{N_c^2}\det(A^{-1})\exp\left[\frac{k^2}{N_c}\,\bar{\mathbf z}^T(\frac{1}{N_c}A^{-1}-\id)\mathbf z\right]\\
&=\frac{N'(N'+1)}{N'+2N' N_c - N_c^2}\exp\left[-k^2\frac
  {(N'-N_c)(|z|^2+|z'|^2)-\sqrt{N'(N'+1)}(z\bar z'+\bar z z')}{N'+2N'
    N_c-N_c^2}\right]. \label{eq:R_Np}
\end{align}
To simplify this expression, let $c=N'-N_c$ and $d=\sqrt{N'(N'+1)}$; these are the prefactors of the first and second terms in the argument of the exponential, absent the denominator. 
Observe that the denominator itself is simply $d^2-c^2$.  
Now notice that $s$ as defined in the statement of the lemma is such that $s/(1+s)=c/d$.
By direct substitution it is easy to show that the prefactor of the exponential in \eqref{eq:R_Np} simplifies to
\begin{align}
\frac{N'(N'+1)}{N'+2N'N_c-N_c^2}=\frac{(1+s)^2}{1+2s}.
\end{align}
Therefore,
\begin{align}
R_{N'}(z,z')
&=\frac{(1+s)^2}{1+2s}\exp\left[-\frac{k^2}{c}\frac{s}{1+2s}\left(s(|z|^2+|z'|^2)-(1+s)(z\bar z'+\bar z z')\right) \right].
\end{align}
As $c=k^2N$, we obtain
\begin{align}
R_{N'}(\sqrt{\tfrac N {2}} z,\sqrt{\tfrac N {2}} z')
&=\frac{(1+s)^2}{1+2s}\exp\left[-\frac{s}{2(1+2s)}\left(s(| z|^2+| z'|^2)-(1+s)( z{\bar z}'+{\bar z}  z')\right) \right],
\end{align}
and thus, finally,
\begin{align}
R_{N'}(\sqrt{\tfrac N {2}} (x+iy),\sqrt{\tfrac N {2}} (x'+iy'))
&=K_s(x,x')K_s(y,y').
\end{align}
Returning to \eqref{eq:QQR} and changing variables $z\to \sqrt{\frac N2}(x+iy)$ yields
\begin{align}
1+\chi^2(\rho,\tau_{N'})
&=\left(\int_{\mathbb{R}} \dd x\,\, \dd {x'} P_{X_1'}(x)\, P_{X_1'}(x') K_s(x,x')\right)
\left(\int_{\mathbb{R}} \dd x\,\, \dd {x'} P_{X_2'}(x)\, P_{X_2'}(x') K_s(x,x')\right).
\end{align}
Appealing to Lemma~\ref{lem:classicalchi2expression} completes the proof. \hfill\qedhere
\end{proof}

\printbibliography[heading=bibintoc,title=References]

\end{document}